\documentclass[10pt, final]{IEEEtran}
\usepackage[utf8x]{inputenc}
\usepackage{epsfig,psfrag,amsmath,float, graphicx}
\usepackage{latexsym,amssymb,amsfonts,cite}

\usepackage{tikz}
\usetikzlibrary{arrows,positioning,plotmarks}
\usepackage{verbatim}
\usepackage{etex}
\usepackage[american]{babel}
\usepackage[T1]{fontenc}
\usepackage{pgfplots}

\IEEEaftertitletext{\vspace{-2\baselineskip}}
{\newtheorem{Thm}{Theorem}}
\newtheorem{Lem}{Lemma}

\newtheorem{Def}{Definition}

\newtheorem{Remark}{Remark}
 
\DeclareMathOperator*{\argmax}{arg\,max}
\DeclareMathOperator{\diag}{diag}
\DeclareMathOperator{\logdet}{log\,det}
\newcommand{\half}{\ensuremath{\frac{1}{2}}}
\DeclareMathOperator{\tr}{tr}
\DeclareMathOperator{\sinr}{SINR}
\DeclareMathOperator{\re}{Re}
\DeclareMathOperator{\im}{Im}
\DeclareMathOperator{\tri}{\triangle}

\title{Improper Gaussian Signaling on the Two-user SISO Interference Channel%
\thanks{Manuscript received Nov. 28, 2011; revised Mar. 2, 2012; accepted May 11, 2011. The associate editor coordinating the review of this paper and approving it for publication was Prof. Murat Torlak.}
\thanks{This work has been performed in the framework of the European research
project SAPHYRE, which is partly funded by the European Union under 
its FP7 ICT Objective 1.1 - The Network of the Future. A conference version of the results is published in \cite{Ho2011}.}
\thanks{Zuleita Ho and Eduard Jorswieck are with the Institut of Communication Technology,
Faculty of Electrical and Computer Engineering,
Dresden University of Technology, Germany.
(\{zuleita.ho, eduard.jorswieck\}@tu-dresden.de)}}
\author{Zuleita K.M. Ho \IEEEmembership{Member,~IEEE} and Eduard Jorswieck \IEEEmembership{Senior Member,~IEEE}}
\begin{document}

\maketitle
\begin{abstract}
On a single-input-single-out (SISO) interference channel (IC), 
conventional non-cooperative strategies encourage players selfishly maximizing their transmit data rates, neglecting the deficit of 
performance caused by and to other players. In the case of proper complex Gaussian noise, the maximum entropy theorem shows that the best-response strategy is to transmit with
proper signals (symmetric complex Gaussian symbols). However, such equilibrium leads to degrees-of-freedom zero due to the
saturation of interference. 

With improper signals (asymmetric complex Gaussian symbols), an extra freedom of optimization is available. In this paper, we study 
the impact of improper signaling on the 2-user SISO IC. We explore the
achievable rate region with non-cooperative strategies by computing a Nash equilibrium of a non-cooperative game with improper signaling.
Then, assuming cooperation between players, we study the achievable rate region of improper signals. 
We propose the usage of improper rank one signals for their simplicity and ease of implementation. Despite their simplicity, rank one signals achieve close to optimal sum rate 
compared to full rank improper signals. 
We characterize the Pareto boundary, the outer-boundary of the achievable rate region, of improper rank one signals with a single real-valued parameter; we compute the closed-form solution of the Pareto boundary with the non-zero-forcing strategies, the maximum sum rate point and the max-min fairness solution with zero-forcing strategies. 
Analysis on the extreme SNR regimes shows that proper signals maximize the wide-band slope of spectral efficiency whereas improper signals
optimize the high-SNR power offset.
\end{abstract}

\section{Introduction}
The characterization and computation of the capacity of the interference channel has been an intriguing and open problem. Although the exact capacity is not known even for the simplest form,
the 2-user SISO IC, inspiring approximations \cite{Tse2007,Etkin2008}, achievable rate regions \cite{Han1981, Sason2004} and capacity regions \cite{Carleial1975, Sato1981, Zhou2010a} and outerbounds are available \cite{ Kramer2004,Tuninetti2008,Motahari2009, Shang2009a}. There are several approaches to tackle this intriguing problem. We give some representable results in the following.

The capacity of the two-user SISO IC is only known for a certain range of channel parameters.
In the weak interference regime where the cross interference gain is much weaker than the direct channel gain, the sum rate capacity is achievable by treating interference as additive noise at the receiver \cite{Shang2009, Motahari2009}, whereas in the strong and very strong interference regime, the interference signal is decoded and subtracted before the desired signal is decoded without any interference \cite{Han1981, Carleial1975, Sato1981,Kramer2004}.  In the mixed interference regime, where one cross interference gain is stronger than direct channel gain and the other link is weaker, the sum rate capacity is shown to be attained by one user decoding interference and the other user treating interference as noise \cite{Motahari2009, Weng2008}.

The deterministic channel approach offers a good approximation of  the sum capacity of interference channel by modeling the input-output relationship of the channel  as a bit-shifting operation \cite{Bresler2008, Cadambe2008, Jafar2010}. The deterministic channel shifts the transmitted bits to a level determined by the signal-to-noise-level of the links. The bits received lower than the noise level is considered lost. At the receivers, the interference signals and target signals undergo a modulo-two addition. 

In the view of game theory, the two users in the SISO IC suffer from the conflict of resources such as power and frequency resources. Assuming no cooperation among
the users, the resource optimization problem can be modeled as a non-cooperative game in power optimization at the transmitters. It is well accounted for that the Nash equilibrium (NE) point, achieved by both users selfishly maximizing their transmission rates and ignoring the interference generated towards the other users, are \emph{not efficient} \cite{Larsson2009}.  In \cite{Leshem2006,Liu2010}, the two users in the SISO IC are allowed to cooperate, or bargain. The bargaining process converges to the  Nash bargaining solution (NBS) which is Pareto efficient. With bargaining, as a form of cooperation, the achievable rate region is enlarged compared to the non-cooperative region.

However, all of the above are limited to the conventional proper signaling. Proper Gaussian signals, or so-called symmetric complex Gaussian signals, have been widely used in communication systems due to its attractive properties.
One of the most important properties includes \emph{the maximum entropy theorem} which states that the differential entropy of 
a complex random variable or random processes with a given second moment is maximized if and only if
the random variable or random processes are proper Gaussian \cite{Neeser1993}.

\subsection{Improper signaling}
Proper Gaussian signals, or the so-called \emph{symmetric complex signaling}, are complex Gaussian scalar variables such that the real and imaginary part of the symbols have
equal power and are independent zero-mean Gaussian variates \cite{Neeser1993,Schreier2010}. If the real and imaginary parts of the complex Gaussian symbols either have unequal power or are correlated, then the symbols are called \emph{improper}. Improper signals have wide applications in signal processing and information theory in \cite{Adali2011,Taubock2010,Schreier2010}  and references therein. Improper signaling techniques are implemented in GSM \cite{Mostafa2003a, Ottersten2005} and 3GPP networks \cite{ST-NXPWirelessFranceCom-Research2009a,ST-NXPWirelessFranceCom-Research2009}.  
In the area of communication theory, with the assumption of proper signals, it is shown in \cite{Host-Madsen2005} 
that the degrees-of-freedom (DOF) of the two-users interference channel is one. The authors show that DOF of $K$-users interference channel
is at most $K/2$. It is shown in \cite{Jafar2009} that by
employing improper signals and the concept of interference alignment, the DOF 1.2 is achievable on a 3-user SISO interference
channel in \cite{Jafar2009} and a DOF of 1.5 on a system of three SISO interfering broadcast channel in \cite{Shin2010}. 

Although the DOF in \cite{Jafar2009}, representing the number of error-free data streams achievable and the slope 
of the sum rate curve in the high SNR regime, is an important measure of performance, a further analysis of the impacts of improper 
signaling on the achievable rate region and optimization of signals should be pursued in order to improve system efficiency, 
such as the NE, the max-min fair operating point and the maximum
sum rate point for finite SNR. In \cite{Park2010}, the max-min fairness solution is studied on the 2-user SISO IC with improper signaling. Illustrated by simulation, 
the max-min fairness solution is improved significantly by improper signaling in all SNR regimes. 
 
The goal of this paper is to examine and study the impact and optimal use of improper signaling in the scenario of two-user SISO IC. To this end, in Section \ref{sec:non_coop}, we study the non-cooperative scenario  where players are selfish and intend to improve their achievable rates by transmitting with improper signaling (which is a larger set and includes proper signaling). The non-cooperative solution depends on local channel state information (CSI) as there is no cooperation among players. The interesting result is that proper signaling is an equilibrium point.  However, this implies that one non-cooperative Nash equilibrium is always proper and thereby not efficient in the interference limited scenarios. 

This motivates the study of a cooperative scenario in Section \ref{sec:coop} in which users  transmit improper signaling to improve various utilities such as achievable rate region, max-min fairness solution and proportional fairness solution. In Thm. \ref{thm:non_zf} and \ref{thm:zf}, we characterize the Pareto boundary, maximum sum rate point and max-min fairness with simple rank one improper signaling. To operate on the Pareto boundary, the users are assumed to have some coordination, hence a cooperative scenario. As illustrated by simulations in Section \ref{sec:sim}, improper signaling improves these utilities in different SNR regimes.  In extreme SNR regimes, we show in Lem. 2 and 3 in Section \ref{sec:extreme_snr} that proper signaling maximizes the wide-band slope of the spectral efficiency whereas the improper signaling maximizes the high SNR power offset. This illustrates that proper signals are more preferable in the noise limited regime and improper signals are more preferable in the interference limited regime.

During the preparation of the final version of this paper, an iterative convex optimization is proposed to solve for the covariance matrices which attain the Pareto boundary of the 2-user SISO IC. Simulation results in \cite{Zeng2012} verify that our proposed MMSE method, which can obtained in closed form, attains the maximum sum rate point despite being rank one.

\subsection{Notations} The $\re(.), \im(.)$ are operators which
return the real and imaginary parts of a complex number. 
The null space of a matrix $\mathbf{A}$ is denoted as $\mathcal{N}(\mathbf{A})$ in which for any $\mathbf{x} \in \mathcal{N}(\mathbf{A})$, $\mathbf{A} \mathbf{x}=\mathbf{0}$\footnote{In this paper, we focus only on the column null space of $\mathbf{A}$. It must not be confused with the row null space of $\mathbf{A}$ which is a set of vectors $\mathbf{x}$ such that $\mathbf{x}^H \mathbf{A}=\mathbf{0}$.}.
The operator $\mathbb{E}(.)$ returns the expectation of a random variable. The quantity $j=\sqrt{-1}$. The function $\log(.)$ has base 2. The function $\ln(.)$ is 
the natural log.
The set $\mathbb{R}$ is the set of real numbers. The operator $\lambda(\mathbf{A})$ returns a vector of eigenvalues of the matrix $\mathbf{A}$.

\section{Channel Model}\label{sec:ch_model}

\tikzstyle{int}=[draw, fill=blue!20, minimum size=2cm]
\tikzstyle{init} = [pin edge={to-,thin,black}]
\begin{figure}
\begin{center}
\resizebox{0.5 \columnwidth}{!}{
\begin{tikzpicture}[node distance=5cm,auto,>=latex']
\node (x1) at (0,3) {$x_1$};
\node (x2) at (0,0) {$x_2$};

\node (y1) at (5,3) {$y_1$};
\node (y2) at (5,0) {$y_2$};

\draw (x1) to node {1} (y1);
\draw (x1) to  (y2);
\draw (x2) to node {1} (y2);
\draw (x2) to  (y1);

\node at (3,2.2) {$h_{12}$};
\node at (3,0.9) {$h_{21}$};

\end{tikzpicture}
}
\caption{The channel model of a two-user SISO IC in the standard form.}
\label{fig:ch_model}
\end{center}
\end{figure}
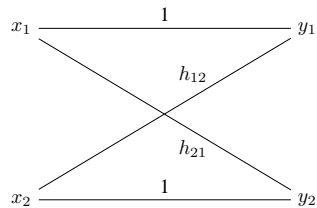

As shown in Fig. \ref{fig:ch_model}, the input-output relationship  of a two-user SISO-IC in the standard form \cite{Carleial1975} is, for $i,k=1,2$, $i \neq k$, given by
\begin{equation}\label{eqt:in-out}
  y_i=  x_i + h_{ik} x_k + n_i.
\end{equation} 
The interference channel coefficients $h_{ik}$ are deterministic complex scalars
 $h_{ik}= \sqrt{g_{ik}} e^{j \phi_{k}}$ with magnitude $\sqrt{g_{ik}}$ and phase $\phi_{ik}$.
The noise $n_i$ is a \emph{proper} complex Gaussian variable with zero mean and variance $N$ and the 
transmit symbols $x_1, x_2$ are complex Gaussian variables which may be \emph{improper}. Before we proceed, we need to formally describe 
propriety.

\begin{Def}[\cite{Adali2011,Taubock2010,Schreier2010}]\label{def:proper}
 A complex random variable $z=z_c + j z_s$ will be called \emph{proper} if its covariance matrix $\mathbf{Q}_z$ is a scaled identity matrix:
\begin{equation}
 \mathbf{Q}_z= \mathbb{E} \left[ \left[\begin{array}{c}
                                   z_c\\
				   z_s
                                  \end{array} \right] \left[ z_c, z_s\right]
 \right].
\end{equation} 
\end{Def}
Note that the covariance matrix  $\mathbf{Q}_z$ is a scaled identity matrix if and only if the power of $z_c$ and $z_s$ are the same and the correlation between $z_c, z_s$ are zero.
By Def. \ref{def:proper}, the transmit symbols $x_i, x_k$ have covariance matrix
\begin{equation}\label{eqt:tx_cov}
 \mathbf{Q}_i=\mathbb{E} \{ \mathbf{x}_i \mathbf{x}_i^T\} = \left[ \begin{array}{cc}
                                                        p_i & \alpha_i\\
\alpha_i & 1-p_i
                                                       \end{array}
\right] P= \mathbf{Q}_i(p_i, \alpha_i)
\end{equation}
with  $\mathbf{x}_i=[\re(x_i), \im(x_i)]^T$ and  transmit power $P$. For simplicity, we assume that the transmit power for both users are the same. The 
parameter $\alpha_i=\mathbb{E}(\re(x_i) \im(x_i))$, which is the correlation between the real and imaginary parts of $x_i$, takes
values between $-\sqrt{p_i(1-p_i)} \leq \alpha_i \leq \sqrt{p_i(1-p_i)}$ to ensure the positive semi-definiteness of the 
transmit covariance matrix.  The variable $x_i$ is proper if and only if $\alpha_i=0$ and $p_i=P/2$ \cite{Adali2011,Taubock2010,Schreier2010}. 

To compute the achievable rate using Shannon's formula, we proceed with the real-valued representation of \eqref{eqt:in-out}
\begin{equation}\label{eqt:in-out-real}
 \mathbf{y}_i= \mathbf{x}_i + \sqrt{g_{ik}} \mathbf{J}(\phi_{ik}) \mathbf{x}_k + \mathbf{n}_i
\end{equation} where $\mathbf{y}_i=[\re(y_i), \im(y_i)]^T$, $\mathbf{n}_i=[\re(n_i), \im(n_i)]^T$ and the channel matrix resembles a rotation matrix
, with $-\pi \leq \phi_{ik} \leq \pi$:
\begin{equation}
 \mathbf{J}(\phi_{ik})= \left[ \begin{array}{cc}
        \cos(\phi_{ik})  &  -\sin(\phi_{ik})\\
\sin(\phi_{ik}) & \cos(\phi_{ik})                         
                                \end{array}
\right].
\end{equation}
Eqt. \eqref{eqt:in-out-real} resembles the conventional input-output relationship on the multiple-input-multiple-output (MIMO)-IC. An achievable rate with 
\emph{real} transmit symbols and transmit covariance matrix $\mathbf{Q}_i \in \mathbb{R}^{2 \times 2}$  is \cite{Goldsmith2003}
\begin{equation}\label{eqt:rate_log_det}
\begin{aligned}
& R_i(\mathbf{Q}_1,\mathbf{Q}_2) = \half \logdet \left(\frac{N}{2}\mathbf{I} + \mathbf{Q}_i + g_{ik} \mathbf{J}(\phi_{ik}) \mathbf{Q}_k \mathbf{J}(\phi_{ik})^T \right)\\
&  - \half \logdet \left(\frac{N}{2}\mathbf{I} + g_{ik} \mathbf{J}(\phi_{ik}) \mathbf{Q}_k \mathbf{J}(\phi_{ik})^T\right)
\end{aligned}
\end{equation}
and the corresponding achievable rate region is $\mathcal{R}$
\begin{equation}\label{eqt:rate_region}
 \mathcal{R}= \bigcup_{\mathbf{Q}_1, \mathbf{Q}_2 \in \mathbb{S}} \bigg(R_1 \big(\mathbf{Q}_1,\mathbf{Q}_2 \big),R_2 \big(\mathbf{Q}_1,\mathbf{Q}_2 \big) \bigg)
\end{equation}
where 
 $\mathbb{S}=\left\{\mathbf{Q} \in \mathbb{R}^{2 \times 2}: \lambda(\mathbf{Q}) \geq 0, \tr \left\{ \mathbf{Q} \right\} \leq P \right\}$
denotes the set of two-by-two positive semi-definite matrices with power less than or equal to $P$.
Using \eqref{eqt:tx_cov}, we can write a more precise characterization of an achievable rate region $\mathcal{R}$:

\begin{equation}\label{eqt:rate_region}
\begin{aligned}
\mathcal{R}&= \bigcup_{(p_i,\alpha_i)\in \mathbb{A}_i} \Bigg(R_1 \bigg(\mathbf{Q}_1(p_1,\alpha_1),\mathbf{Q}_2(p_2,\alpha_2) \bigg), \\
& \hspace{2cm} R_2 \bigg(\mathbf{Q}_1(p_1,\alpha_1),\mathbf{Q}_2(p_2,\alpha_2) \bigg) \Bigg)
\end{aligned}
\end{equation} 
where for $i=1,2$,
 \begin{equation*}
\begin{aligned}
 \mathbb{A}_i &=\left\{ (p_i, \alpha_i) : 0 \leq p_i \leq 1,\right.\\
& \hspace{2cm} \left. -\sqrt{p_i(1-p_i)} \leq \alpha_i \leq \sqrt{p_i(1-p_i)} \right\}.
\end{aligned}
\end{equation*}
\begin{Def}
 The Pareto boundary $\mathcal{B}$ of rate region $\mathcal{R}$, 
achievable by covariance matrices $\mathbf{Q}_1, \mathbf{Q}_2 \in \mathbb{S}$,
 is defined as the boundary points of $\mathcal{R}$. If $(r_1,r_2) \in \mathcal{B}$, then there
does not exist a point $(r_1,r_2')$ or $(r_1',r_2)$ such that $r_1'>r_1$ or $r_2'>r_2$.
The rates are $r_i=R_i(\mathbf{Q}_1,\mathbf{Q}_2)$ and $r'_i=R_i(\mathbf{Q}'_1,\mathbf{Q}'_2)$ for some
$\mathbf{Q}_i,\mathbf{Q}'_i \in \mathbb{S}$.
Consequently, it is impossible to increase one user's rate without decreasing the others.
\end{Def}

In the following sections, we study various non-cooperative and cooperative strategies in the achievable rate region $\mathcal{R}$.

\section{Non-cooperative Solution}\label{sec:non_coop}
With no cooperation between users, each user maximizes its own transmit data rate, neglecting the 
possible deficit of the other user's performance. The receivers are assumed to have channel state information of the channel from its corresponding transmitter in order to convert any interference channel to the standard form of the channel input-output relationship as in \eqref{eqt:in-out}. As the direct channel gain is normalized to one and the transmitters are not interested in interference management, no channel state information is required at the transmitters. 
The behavior of this setting is best formulated 
as a non-cooperative game:
\begin{Def}
 The two-user SISO IC non-cooperative game $\mathcal{G}$ with improper signaling is defined as a tuple 
$\mathcal{G}=(\mathbb{N}, \mathbb{Q}, \mathbb{U} ) 
$ where $\mathbb{N}=\{1,2 \}$ is the set of players; $\mathbb{Q}= \mathbb{S} \times \mathbb{S}$
 is the set of strategies and $\mathbb{U}=\{R_i, i\in \mathbb{N}\}$ is the set of utilities.
\end{Def}
\begin{Lem}\label{lem:nash}
One Nash Equilibrium of $\mathcal{G}$ is attained by the dominant strategies $\mathbf{Q}_i=\frac{P}{2} \mathbf{I}$, the proper Gaussian symbols.
\end{Lem}
\begin{proof}
See Appendix \ref{app:NE}.
\end{proof}

Lem. \ref{lem:nash} illustrates that even if the assumption of proper transmit signals is relaxed, proper Gaussian signaling remains as an equilibrium point. 
If one user employs proper signals, the other user, in order to maximize its data rate, is best to employ proper  signals also. Once both users employ 
proper signals, neither has an incentive to deviate from this operating point which prevents the system to potentially achieve a more efficient operating point.

By Lem. \ref{lem:nash}, one Nash Equilibrium of $\mathcal{G}$ is given by
\begin{equation}\label{eqt:nash}
 R^{NE}=\left( R_1\left(\frac{P}{2}\mathbf{I},\frac{P}{2}\mathbf{I}\right),R_2\left(\frac{P}{2}\mathbf{I},\frac{P}{2}\mathbf{I}\right) \right).
\end{equation}
The single user points, achieved by proper signaling, are $R^{SU1}=\left( \log(1+\frac{P}{N}),0\right)$ and $R^{SU2}=\left(0, \log(1+\frac{P}{N})\right)$. Note that all the aforementioned operating points with proper signaling are in the achievable rate region $\mathcal{R}$ defined in \eqref{eqt:rate_region}: $R^{NE}, R^{SU1}, R^{SU2} \in \mathcal{R}$. 

However, the non-cooperative equilibrium operating point is in general not efficient, especially in the high SNR regime. The transmit strategy 
at the NE contributes strong interference to each user and saturates the sum rate performance. To improve the efficiency of the transmit strategies, we study the performance of cooperative solutions, allowing the transmitters to optimize the transmit
covariance matrices with the assumption of improper signaling.
It can be shown that  simple improper transmit strategies can restore the DOF and 
contribute to a significantly larger achievable rate region compared to  proper signaling. 

\section{Cooperative Solutions}\label{sec:coop}
To provide a better understanding of the impact of improper signaling, we study rate regions of various cooperative solutions. The transmitter-receiver pairs are willing to optimize a common metric or utility function. In particular,  both pairs would like to operate on the Pareto boundary. To this end, we compute the characterization of the Pareto boundary with single-value parameters which can be seen as a centralized approach. For decentralized implementation, to achieve operating points on the convex hull of the achievable rate region, one can adapt the methods similar to \cite{Shi2011}.
For the ease of analysis and illustration, we assume rank one transmit covariance matrices 
$ \mathbf{Q}_i=\mathbf{q}_i \mathbf{q}_i^T P$
where
$\mathbf{q}_i= [\cos(\tau_i), \sin(\tau_i)]^T, 0 \leq \tau_i \leq \pi$.
The input-output relationship from \eqref{eqt:in-out-real} is rewritten to the following:
\begin{equation}
 \begin{aligned}
  y_1 &= \mathbf{v}_1^T \mathbf{q}_1 d_1 + \sqrt{g_{12}}\mathbf{v}_1^T \mathbf{J}(\phi_{12}) \mathbf{q}_2 d_2 + \mathbf{v}_1^T \mathbf{n}_1\\
  y_2 &= \sqrt{g_{21}} \mathbf{v}_2^T \mathbf{J}(\phi_{21}) \mathbf{q}_1 d_1 + \mathbf{v}_2^T \mathbf{q}_2 d_2 + \mathbf{v}_2^T \mathbf{n}_2\\
 \end{aligned}
\end{equation}
The variables $d_i \in \mathbb{R}$ are the data symbols and the transmit symbols are $\mathbf{x}_i=\mathbf{q}_i d_i$ and $\mathbf{v}_i$ are the receive beamforming vectors.
An achievable rate region of rank-1 transmit covariance matrices is denoted as $\mathcal{R}^{one}$:
\begin{equation}\label{eqt:rank1_char}
\begin{aligned}
\mathcal{R}^{one}&=  \bigcup_{\substack{ i=1,2, \\ 0 \leq \tau_i \leq \pi}} \Bigg\{ \bigg(R_1 \big(\mathbf{q}_1\mathbf{q}_1^T P, \mathbf{q}_2 \mathbf{q}_2^T P \big), \\
& \hspace{3cm} R_2 \big(\mathbf{q}_1\mathbf{q}_1^T P,\mathbf{q}_2 \mathbf{q}_2^T P \big) \bigg) \Bigg\}
\end{aligned}
\end{equation}
with notation $\mathbf{q}_i=\mathbf{q}_i(\tau_i)$ and $\mathcal{R}^{one} \subseteq \mathcal{R}$.
We analyze the performance of the rate region with two separate cases: with ZF strategies $\mathcal{R}^{zf}$ and non-ZF strategies $\mathcal{R}^{nzf}$.
\begin{equation}
 \mathcal{R}^{one}= \mathcal{R}^{nzf} \cup \mathcal{R}^{zf}.
\end{equation}
The non-ZF strategies provide a more general setting in which MMSE receivers are employed and perform better than ZF strategies in finite SNR regimes. However, the ZF strategies provide
a simpler and easier implementation than non-ZF strategies as shown in later sections.

\subsection{Non-ZF strategies}\label{sec:non_zf}
In this subsection, we assume that the transmit and receive beamforming vectors do not jointly null out interference. The receivers 
are assumed to employ a minimum-mean-square-error receiver $\mathbf{v}_i$ which results in the Shannon rate formula.
Let $\mathbf{A}= \frac{N}{2}\mathbf{I}+  g_{ik} P \mathbf{J}(\phi_{ik}) \mathbf{q}_k \mathbf{q}_k^T \mathbf{J}(\phi_{ik})^T$.
Apply $\det(\mathbf{A}+\mathbf{u}\mathbf{v}^T)= (1+ \mathbf{v}^T \mathbf{A}^{-1} \mathbf{u}) \det(\mathbf{A})$. 
We rewrite the rate of user $i$ as 
$R_i= \half \log \left( 1+ \mathbf{q}_i^T \mathbf{A}^{-1} \mathbf{q}_i P\right).
$
Using the matrix inversion lemma on $\mathbf{A}$,
we obtain
\begin{equation}\label{eqt:rate_bf}
 R_i= \half \log \left( 1+ 2 \gamma - \frac{ 4 g_{ik}\gamma^2 \left(\mathbf{q}_k^T\mathbf{J}(\phi_{ik})^T \mathbf{q}_i \right)^2}{1+ 2 g_{ik} \gamma}\right)
\end{equation} 
where 
$ \gamma=\frac{P}{N}
$ is the signal-to-noise-ratio of both users.
\begin{Thm}\label{thm:non_zf}
The achievable rate region of non-ZF beamforming strategies can be completely characterized by a single real-valued parameter $\tri \tau$ as
\begin{equation}\label{eqt:ach_region_nzf}
 \mathcal{R}^{nzf}=\bigcup_{0 \leq \tri \tau \leq \pi} \left( R_1(\tri \tau), R_2(\tri \tau)\right).
\end{equation}
where $\tri \tau=\tau_2-\tau_1$ and $0 \leq \tri \tau \leq \pi$.
 For some achievable $\sinr_1$ attained by the transmit vector pair $(\mathbf{q}_1,\mathbf{q}_2)$, the $\sinr_2$  is function of $\sinr_1$ only:
\begin{equation}\label{eqt:sinrk}
\begin{aligned}
&\sinr_2(\sinr_1)\\
&=  2 \gamma- \frac{4 g_{21} \gamma^2}{ 1+2 g_{21} \gamma} \left( \sqrt{D} \cos(\bar{\phi})+ (-1)^b \sin(\bar{\phi}) \sqrt{1-D}\right)^2 
\end{aligned}
\end{equation}
where $ D=\frac{1}{4 g_{12} \gamma^2} \left( 2 \gamma -\sinr_1 \right) \left( 1+ 2 g_{12} \gamma\right)$, $\bar{\phi}=\phi_{12}+\phi_{21}$ and
$b$ is of binary values: 0 or 1.
The Pareto optimality is equivalent to maximizing the $\sinr_2(\sinr_1)$ in Eqt. \eqref{eqt:sinrk} given the value of $\sinr_1$.
 The optimum values of $b$ are:
\begin{equation}
b=\left\{ \begin{array}{cc}
   0 & \hspace{1cm} \pi/2 \leq \bar{\phi} \leq \pi, 3\pi/2 \leq \bar{\phi} \leq 2 \pi,\\
  1 & \hspace{1cm} 0 \leq \bar{\phi} \leq \pi/2, \pi \leq \bar{\phi} \leq  3\pi/2.
  \end{array}
\right.
\end{equation}
\end{Thm}
\vspace{0.3cm}
\begin{proof}
See Appendix \ref{app:sinr_sinr}.
\end{proof}
Note that  for a positive real value of $\sinr_2$, the value of $D$ is required to be in the range $0 \leq D \leq 1$. This bound translates to
a bound of the value of $\sinr_1$: 
$\frac{2 \gamma}{1+ 2 g_{12} \gamma} \leq \sinr_1 \leq 2 \gamma$. Using the result from Thm. \ref{thm:non_zf}, the maximum achievable $\sinr_2$ can be computed in closed form given the target $\sinr_1$ and CSI.

\subsection{Zero-forcing strategies}\label{sec:zf}
In this section, we study an achievable region and the maximum sum rate point by ZF strategies. 
Instead of the MMSE receiver used in Section \ref{sec:non_zf}, a receiver $\mathbf{v}_i=\left[\cos(\theta_i), \sin(\theta_i)\right]^T$ is employed.

The ZF criteria is therefore, for $i=1,2$,
$\mathbf{v}_i^T \mathbf{J}(\phi_{ik}) \mathbf{q}_k=0
$ which can be simplified to
$\cos(-\theta_i+\phi_{ik}+\tau_k)=0
$. The ZF receiver  compensates the phase shift by the interference channel, 
\begin{equation}
 \theta_i=\phi_{ik}+\tau_k-\frac{\pi}{2}.
\end{equation}
With no interference and normalized noise power, the resulting SINR is the desired signal energy:
\begin{equation*}
\begin{aligned}
 \sinr_i&=|\mathbf{v}_i^T \mathbf{q}_i|^2 \gamma\\
&= \cos(\tau_i-\theta_i)^2 \gamma= \cos^2\left(\tau_i-\phi_{ik}-\tau_k+\frac{\pi}{2}\right) \gamma.
\end{aligned}
\end{equation*}
In the following, we obtain the characterization of $\mathcal{R}^{zf}$, the maximum sum rate point in $\mathcal{R}^{zf}$, $R^{zf}$, and the max-min fairness solution.
\begin{Thm}\label{thm:zf}
  The ZF rate region $\mathcal{R}^{zf}$ is parametrized by a single parameter $0 \leq \tri \tau \leq \pi$:
\begin{equation}
\begin{aligned}
 \mathcal{R}^{zf}& =\bigcup_{0 \leq \tri \tau \leq \pi} \bigg( \half \log(1+\sinr_1(\tri \tau)), \\
& \hspace{3cm} \half \log(1+\sinr_2(\tri \tau))\bigg)
\end{aligned}
\end{equation}
where 
$\sinr_1(\tri \tau)=\sin^2(\tri \tau +\phi_{12}) \gamma$ and $\sinr_2(\tri \tau)=\sin^2(\tri \tau - \phi_{21}) \gamma$.  
The maximum sum rate point $R^{zf}$ is attained by
\begin{equation}
\begin{aligned}
\tri \tau &= - \half \tri \phi, \frac{\pi}{2}- \half \tri \phi \\
& \hspace{1cm} \mbox{ or } \frac{1}{2} \cos^{-1} \left( \frac{\gamma+2}{\gamma} \cos(\bar{\phi})\right)- \frac{1}{2} \tri \phi.
\end{aligned}
\end{equation} The max-min fairness solution is attained by 
\begin{equation}
 \tri \tau= - \half \tri \phi, \frac{\pi}{2}- \half \tri \phi
\end{equation}
 with  $\tri \phi=\phi_{12}-\phi_{21}$.
\end{Thm}
\begin{proof}
 see Appendix \ref{app:lem_zf_rate}.
\end{proof}
\begin{Remark}
The maximum sum rate point and the max-min fairness solution assuming ZF solution are derived in closed form. Although ZF type solutions are suboptimal in finite SNR, they attain the maximum degree of freedom in high SNR regime. 
\end{Remark}
\begin{Remark}
 The regions $\mathcal{R}^{nzf}$ and $\mathcal{R}^{zf}$ are both characterized by a real-valued parameter $\tri \tau$, with range $0 \leq \tri \tau \leq \pi$ in Thm. 
 \ref{thm:non_zf} and \ref{thm:zf} respectively. Hence, we can characterize $\mathcal{R}^{one}$, which is the union of both, by the same parameter which simplifies the original characterization in Eqt. \eqref{eqt:rank1_char}.
\end{Remark}

In Section \ref{sec:coop}, we study the achievable rate region with improper rank one strategies whereas in Section \ref{sec:non_coop} we study the achievable rate region of proper signaling, including the NE. In the setting of proper signaling, NE outperforms other strategies in the noise limited regime whereas ZF strategies are optimal in the interference limited regime. The
intriguing problem is to examine whether the same analogy holds in the setting of improper signaling. 

\section{Optimal signaling in extreme SNR regimes}\label{sec:extreme_snr}
In this section, we study the performance of proper and improper signaling in extreme SNR regimes where the transmit power budgets of both users are either very small in the low SNR regime (noise limited) or very  large in the high SNR regime (interference limited).

\subsection{The low SNR regime}
The spectral efficiency is defined as the number of transmission bits per time and frequency channel use \cite{Verdu2002}, which has been widely used to analyze
the performance of different transmission strategies in the low SNR regimes. As pointed out in \cite{Verdu2002}, different transmit strategies converge
to the same $\left(\frac{E_b}{N_0} \right)_{min}$, minimum energy required for reliable data transmission, but give a different first-order growth,
\emph{the wide-band slope of the spectral efficiency}. 
\begin{equation}\label{eqt:s0}
 S_0= \frac{2 (\dot{R})^2}{- \ddot{R}} \frac{1}{10 \log_{10} 2}
\end{equation}
where $\dot{R}, \ddot{R}$ are the first and second derivatives of the sum rate function $R=R_1+R_2$ at $\gamma=0$.
Our goal is to maximize $S_0$ by varying $\mathbf{Q}_1,\mathbf{Q}_2$. 

\begin{Lem}\label{lem:s0}
 The $\left(\frac{E_b}{N_0} \right)_{min}$ is independent of $\mathbf{Q}_1, \mathbf{Q}_2$. The wide-band slope of spectral efficiency $S_0$ defined in Eqt. \eqref{eqt:s0} is maximized by scaled identity covariance matrices in the 
low SNR regime. In other words, proper Gaussian signals are optimal in terms of first-order growth of the spectral efficiency in the low SNR regime.
\end{Lem}
\begin{proof}
 See Appendix \ref{app:s0}.
\end{proof}
Lem. \ref{lem:s0} confirms our intuition: in the noise limited scenario, the sum of received interference and noise is dominated by the noise power. By assumption, the noise is proper and it results in the propriety of the optimal transmit strategy.

\subsection{The high SNR regime}
Although a proper  signal maximizes the wide-band slope of spectral efficiency, it does not
allow interference nulling. This leads to a high SNR slope of zero. 
It is the aim of this section to study the optimal transmit strategies in the high SNR regime. As studied in \cite{Lozano2005}, the maximum DOF, 
or the slope of the maximum sum rate versus SNR, $S_{\infty}(\mathbf{Q}_1,\mathbf{Q}_2)$, of the two-user SISO-IC is one. 
Hence, in the high SNR regime, a good transmit strategy must be able to null out interference.
For the ease of notation, denote $\mathbf{Q}=\left(\mathbf{Q}_1,\mathbf{Q}_2 \right)$. 
The high-SNR slope achieved by transmit covariance matrices $\mathbf{Q}$ in bits/sec/Hz/(3 dB) is defined as \cite{Lozano2005}
\begin{equation}
 S_{\infty}(\mathbf{Q})=\lim_{\gamma \rightarrow \infty} \frac{R_1(\mathbf{Q})+ R_2(\mathbf{Q})}{\log_2 \gamma}.
\end{equation}
The high-SNR power offset is, for $S_{\infty}>0$,
\begin{equation}
 L_{\infty}(\mathbf{Q})=\lim_{\gamma \rightarrow \infty} \left( \log_2 \gamma - \frac{R_1(\gamma, \mathbf{Q})+R_2(\gamma, \mathbf{Q})}{S_{\infty}(\mathbf{Q})}\right)
\end{equation}
where $R_i(\gamma, \mathbf{Q})$ is the rate achieved with SNR $\gamma$ and transmit covariance matrices $\mathbf{Q}$ in Eqt. \eqref{eqt:rate_log_det}. 

In the following, we compute the high-SNR slope and the power offset with rank one transmit covariance matrices:
\begin{Lem}\label{lem:power_offset}
 The high-SNR slope, DOF, of a two-user SISO-IC with transmit beamforming vectors $\mathbf{q}_1,\mathbf{q}_2$, i.e. $\mathbf{Q}_i=\mathbf{q}_i \mathbf{q}_i^T P$, is $S_{\infty}=1$. The high-SNR power offset $L_{\infty}$ is a function of $\tri \tau$ only, where $\mathbf{q}_i=[\cos(\tau_i), \sin(\tau_i)]^T$ 
and $\tri \tau=\tau_2-\tau_1$. The optimal high-SNR power offset is
\begin{equation}
 \max_{\tri \tau} L_{\infty}= \left\{ \begin{array}{cc}
                     -1 -\log \sin^2 (\half \bar{\phi}) & \mbox{ if } \cos(\bar{\phi})<0\\
		     -1 -\log \cos^2 (\half \bar{\phi}) & \mbox{ if } \cos(\bar{\phi})\geq 0
                    \end{array}
 \right.
\end{equation}
 where $\bar{\phi}=\phi_{12}+\phi_{21}$ and the optimal $\tri \tau$ is 
\begin{equation}
\tri \tau = \left\{ \begin{array}{cc}
                     -\half \tri \phi & \mbox{ if } \cos(\bar{\phi})<0\\
		     \frac{\pi}{2} -\half \tri \phi & \mbox{ if } \cos(\bar{\phi})\geq 0
                    \end{array}
 \right.
\end{equation}

\end{Lem}
\begin{proof}
 see Appendix \ref{app:power_offset}.
\end{proof}
 As predicted by Lem. \ref{lem:power_offset}, the high-SNR power offset of rank one strategies are negative as it performs better than the curve $\log(\gamma)$. More details are given in Section \ref{sec:extreme}. 

\section{Results and Discussion}\label{sec:sim}
 In this section, we numerically evaluate an achievable rate region of the 2-user SISO IC in various SNR regime. 
 Given a particular channel realization, the achievable rate region of proper signaling is compared to the achievable rate region of improper 
 signaling. We show in the following that improper signaling provides significant gains in the achievable rate region compared to proper signaling with time sharing, including remarkable improvement in terms of max-min fairness and proportional fairness. For completeness, we include here the definitions of max-min fairness operating point  $\left(R_1^{(mm)},R_2^{(mm)}\right)$:
\begin{equation}
 R_1^{(mm)}=R_2^{(mm)}=\max \min_l R_l^{(mm)}
\end{equation} and the definition of the proportional fairness operating point $(R_1^{(pf)}, R_2^{(pf)})$\footnote{Sometimes proportional fairness point is defined as $r_1 r_2$, without the fall-back of NE point.}
\begin{equation}
 \left(R_1^{(pf)}, R_2^{(pf)} \right)= \argmax_{(r_1,r_2) \in \mathcal{R}} \; \left( r_1- R_1^{NE} \right) \left( r_2- R_2^{NE}\right)
\end{equation} where $\mathcal{R}$ is the achievable rate region of improper signaling \eqref{eqt:rate_region} and  $R_1^{NE}$ and $R_2^{NE}$ describe the rates of transmitter 1 and 2 at the Nash Equilibrium \eqref{eqt:nash}. Intuitively, at the max-min fairness operating point, both users enjoy the same maximum possible rate in the achievable rate region whereas at proportional fairness operating point, the rate point maximizes the product of improvement over the threat point (NE). We observe that with improper signaling instead of proper signaling, the improvement is particularly significant when the channel is asymmetric in the sense that the interference channel towards one player is stronger, e.g. $\mathbb{E} |h_{12}|^2< \mathbb{E} |h_{21}|^2$.
In Fig. \ref{fig:ach_rate_reg_low}-\ref{fig:ach_rate_reg_high}, we show an achievable rate region of two users with  $2\mathbb{E} |h_{12}|^2 = \mathbb{E} |h_{21}|^2$ and system SNR, defined as $P/N$, is 0dB and 10dB respectively.

As illustrated in Fig.\ref{fig:ach_rate_reg_low}-\ref{fig:ach_rate_reg_high}, the Nash Equilibrium (light blue triangle) is severely inefficient and the inefficiency of NE is more pronounced in the high SNR scenario. If time sharing between the NE and the single user points is allowed, we obtain the achievable rate region of proper signaling (plotted as a light blue curve). If time-sharing between the single user points is included also, we may obtain a larger region, depending on the channel fading coefficients (plotted as a dashed-dotted line). To illustrate the improvement of max-min fairness, we can draw a straight line with slope one in the achievable rate region (plotted as a dotted line). The intersection of such line and the boundary of the rate region with improper signaling (plotted as a solid black line) gives the max-min fairness operating point. The max-min fairness point with proper signals is given by the intersection of such line with the boundary of the achievable rate region of proper signals (plotted as a light blue line). The improvement of max-min fairness can be illustrated by the gap between the aforementioned points, marked by a double arrow in Fig.\ref{fig:ach_rate_reg_low}-\ref{fig:ach_rate_reg_high}.

If one selects a point on the boundary of the achievable rate region by improper signaling (solid black line) and creates a rectangle with this point as the vertex and the opposite vertex with the NE, then the point corresponding to the rectangle with the largest area is the proportional fairness operating point.  The improvement of proportional fairness is
illustrated by a double head arrow in Fig.\ref{fig:ach_rate_reg_low}-\ref{fig:ach_rate_reg_high}. The improvement is significant.  
Moreover, the improvement of rate region of improper signaling from proper signaling with time-sharing can be illustrated by the blue shaded area in Fig.\ref{fig:ach_rate_reg_low}-\ref{fig:ach_rate_reg_high}. The improvement is substantial and becomes more significant when the SNR decreases.

As shown in Fig.  \ref{fig:ach_rate_reg_low}-\ref{fig:ach_rate_reg_high}, the Pareto boundary of the rank one improper signaling schemes, in particular, non-ZF schemes and ZF schemes, are plotted in red and blue curve respectively. Note that the achievable rate region of the non-ZF scheme always includes the ZF one because ZF has the restriction of interference nulling. As the SNR increases, the ZF region will approach the non-ZF region as ZF schemes are optimal in high SNR regimes.

\begin{figure}
\begin{center}
\includegraphics[width=\columnwidth,keepaspectratio]{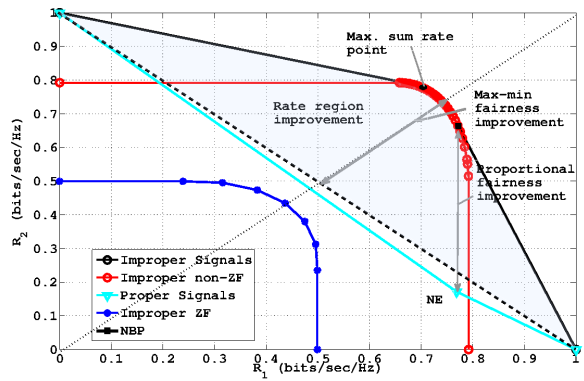}
\caption{The achievable rate region of proper  Gaussian signaling versus improper  Gaussian signaling at low SNR, 0dB. The interference channel from Tx 1 to 2 is twice as strong as the channel from Tx 2 to 1. The achievable rate region improvement by improper signaling is shaded blue. The max-min fairness solution and proportional fairness solution are improved significantly by improper signaling.\label{fig:ach_rate_reg_low}}
\end{center}
\end{figure}

\begin{figure}
\begin{center}
\includegraphics[width=\columnwidth,keepaspectratio]{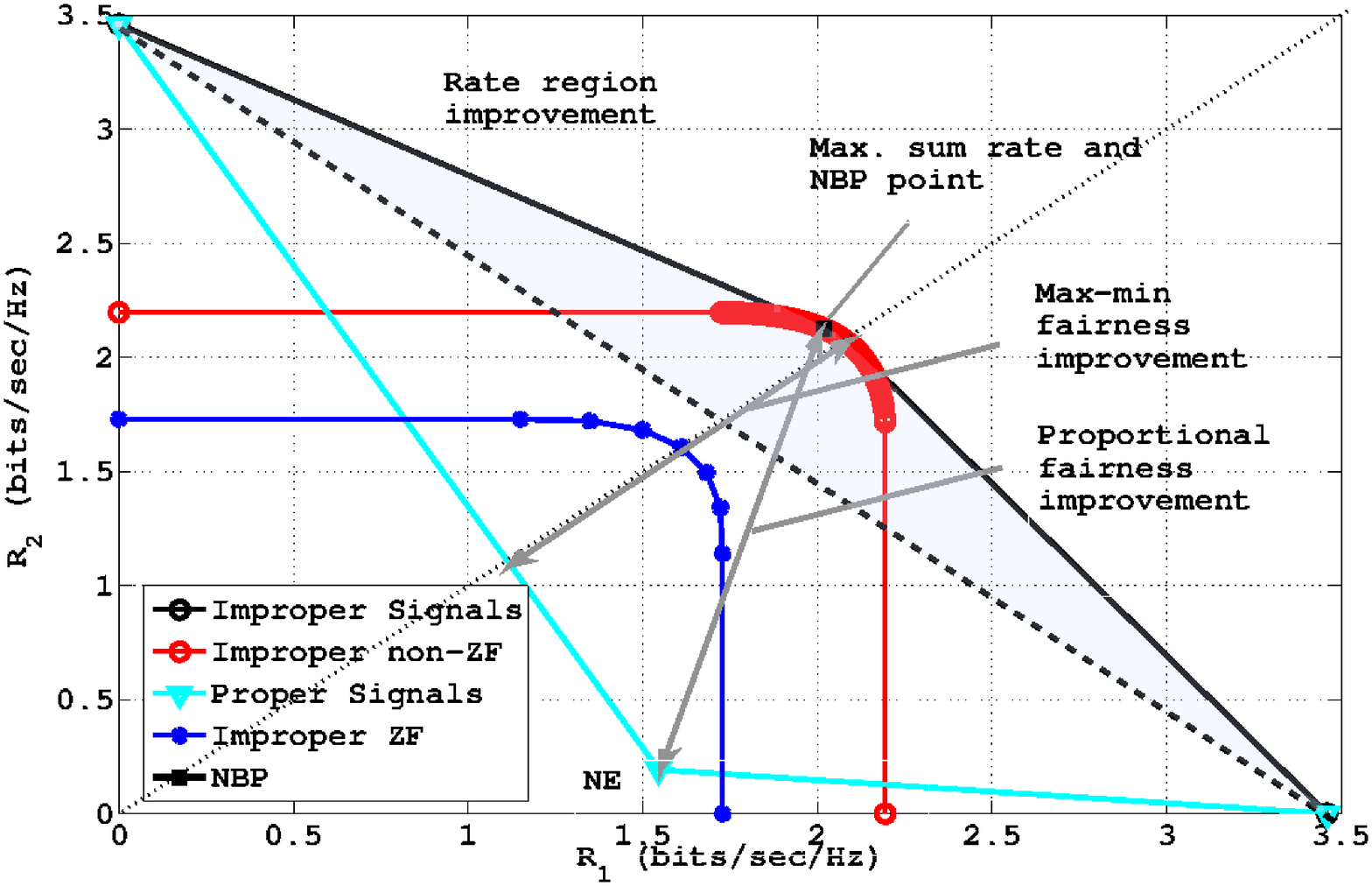}
\caption{The achievable rate region of proper  Gaussian signaling versus improper  Gaussian signaling at medium SNR, 10dB. The interference channel from Tx 1 to 2 is twice as strong as the channel from Tx 2 to 1. The achievable rate region improvement by improper signaling is shaded blue. The max-min fairness solution and proportional fairness solution are improved significantly by improper signaling.\label{fig:ach_rate_reg_high}}
\end{center}
\end{figure}

\subsection{Extreme SNR regimes}\label{sec:extreme}
In Fig. \ref{fig:rate_snr}, the maximum achievable rates of various strategies are plotted over the system SNR. It is encouraging to see
that rank one strategies, despite its simple design, achieve a rate very closed to the general improper signals. In Fig. \ref{fig:rate_snr}, the two curves
are overlapping. Comparing the performance
between the ZF strategies and the proper signaling, we observe that the proper signaling achieves a better sum rate
than ZF strategies in the noise-limited regime whereas ZF strategies achieve a better sum rate than the proper signaling in the interference-limited regime. This observation agrees with the performance of ZF strategies and proper-signaling (the so called maximum-ratio-transmission) in
the MIMO-IC.

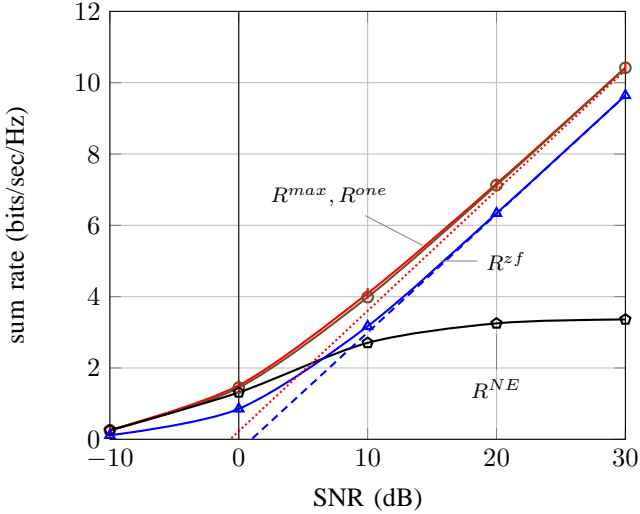
\begin{figure}
\begin{center}

\begin{tikzpicture}[scale=1]
\begin{axis}[xmin=-10,xmax=30,ymin=0,ymax=12,ylabel={sum rate (bits/sec/Hz)},xlabel={SNR (dB)}, legend style={
font=\footnotesize,
}, grid=major]
\tikzstyle{every pin}=[
font=\footnotesize]


\addplot[black] coordinates{
(0,14)
(0,0)
};

\addplot+[smooth,mark=|, thick] coordinates {
(-10, 0.2592)
(0,1.5124)
(10,4.0919)
(20,7.1670)
(30,10.4227)
(40, 13.7387)
}; 
\node[pin=135:{$R^{max},R^{one}$}] at (axis cs:15,5.3) {}; 

\addplot+[smooth, mark=o, thick] coordinates {
(-10, 0.2554)
(0,1.4487)
(10,3.9862)
(20,7.1259)
(30,10.4178)
(40, 13.7387)
}; 

\addplot+[smooth, mark=pentagon, thick] coordinates {
(-10,  0.2523)
(0,1.3146)
(10,2.7047)
(20,3.2505)
(30,3.3616)
(40, 3.3502)
}; 
\node[pin=below:{$R^{NE}$}] at (axis cs:20,3.2) {};     

\addplot+[smooth, mark=triangle, thick] coordinates {
(-10, 0.1129)
(0,0.8546 )
(10,3.1710)
(20, 6.3379)
(30,9.6433)
(40,12.96)
}; 
\node[pin=right:{$R^{zf}$}] at (axis cs:15,5) {}; 

\addplot[thick, densely dashed,blue] coordinates{
(40, 12.96) 
(0.9874,0)
}; 

%
\addplot[thick, densely dotted, red] coordinates{
(40, 13.73)  
(-0.6778,0)
}; 
\end{axis}
\end{tikzpicture}
\caption{The achievable rate of proper Gaussian signaling versus improper  Gaussian signaling at various SNR. The maximum sum rate of improper signaling and rank one improper strategies are plotted with circles and pluses and are overlapping. The dotted line is the asymptote at SNR 30dB which extends and intercept with the x-axis at the high-SNR power offset. In the case of improper signaling, the high-SNR power offset is negative, indicating that it performs better than the reference curve. For ZF strategies, the curve is plotted with triangles and the dashed line is the asymptote at SNR 30dB. The rate achieved by NE is plotted in pentagons and it saturates in high SNR.\label{fig:rate_snr}}

\end{center}
\end{figure}

\section{Conclusion}
We study the impact of achievable rate regions when the assumption of proper Gaussian symbols is relaxed, allowing
the real and imaginary parts of the complex Gaussian symbols to be correlated and with different power. We first explore the
achievable rate region with non-cooperative strategies. We prove that even with improper signaling, one NE remains to be attained
by proper signals which leads to low system efficiency. Then we investigate cooperative strategies by employing improper signals.
We characterize the achievable rate region of improper rank one signals with a single real-valued parameter. In the scope of
rank one strategies, we distinguish two sub-region: non-ZF strategies and ZF strategies. We derive the closed-form solution to 
the Pareto boundary of the non-ZF strategies and the maximum sum rate point, the maxmin fair solution with the ZF strategies. 
Analysis on the extreme SNR regimes shows that proper signals maximize the wide-band slope of spectral efficiency whereas improper signals
maximize the high-SNR power offset.

\appendices 
\section{Proof of Lemma \ref{lem:nash}}\label{app:NE}

Given transmit covariance matrix $\mathbf{Q}_k$ from player $k$,  player $i$ maximizes its rate  in  \eqref{eqt:rate_log_det}, 
$R_i(\mathbf{Q}_i,\mathbf{Q}_j)=\half \logdet \left( \mathbf{Q}_i + \mathbf{A} \right)- \half \logdet \mathbf{A}$ 
where $\mathbf{A}= \frac{N}{2}\mathbf{I} + g_{ik} \mathbf{J}(\phi_{ik}) \mathbf{Q}_k \mathbf{J}(\phi_{ik})^T$. 
We define the best-response function $BR_i(\mathbf{Q}_j)$ of player $i$, taking parameter of the transmit strategy $\mathbf{Q}_j$ of player $j$ and returning the 
best-response strategy of player $i$ such that $R_i$ is maximized:
\begin{equation}
 BR_i(\mathbf{Q}_j)= \argmax_{\mathbf{Q}_i} R_i(\mathbf{Q}_i,\mathbf{Q}_j).
\end{equation}
We can see that if  $\mathbf{Q}_j=\frac{P}{2} \mathbf{I}$, then the matrix $\mathbf{A}$ in $R_i$ is diagonal and the best-response strategy $\mathbf{Q}_i=\frac{P}{2} \mathbf{I}$.  In other words, proper Gaussian signaling is an equilibrium point.

\section{Proof of Theorem \ref{thm:non_zf}}\label{app:sinr_sinr}
Note that the product $\mathbf{q}_k^T\mathbf{J}(\phi_{ik})^T \mathbf{q}_i= \cos(\phi_{ik}+\tau_k-\tau_i) $
is a function of $\tau_k-\tau_i$. Define
$ \tri \tau= \tau_2-\tau_1, -\pi \leq \triangle \tau \leq \pi
$. The achievable rates are functions of $\tri \tau$ only, regardless of the values of individual transmit beamforming vector $\mathbf{q}_i$:
\begin{equation*}\left\{
\begin{aligned}
R_1(\tri \tau)&= \half \log \left( 1+ 2 \gamma - \frac{ 4 g_{12} \gamma^2 \cos^2(\phi_{12}+ \tri \tau)}{1+ 2 g_{12} \gamma}\right)\\
R_2(\tri \tau)&= \half \log \left(1+ 2 \gamma - \frac{ 4 g_{21} \gamma^2 \cos^2(\phi_{21}-\tri \tau)}{1+ 2 g_{21} \gamma} \right)
\end{aligned} \right.
\end{equation*}

Since $R_1(\tri \tau), R_2(\tri \tau)$ relates to $\tri \tau$ only through the function $\cos(.)$ and the fact that
$\cos(\theta + \pi)^2= (-\cos(\theta))^2=\cos^2(\theta)$, we can reduce the range of $\tri \tau$ to
$0 \leq \tri \tau \leq \pi$.
Hence, the achievable rate region of non-ZF beamforming strategies is defined as
$ \mathcal{R}^{nzf}=\bigcup_{0 \leq \tri \tau \leq \pi} \left( R_1(\tri \tau), R_2(\tri \tau)\right)$.

Now, to obtain the Pareto boundary characterization of $\mathcal{R}^{nzf}$, we write $\sinr_2$ as a function of $\sinr_1$.
The Pareto boundary of $\mathcal{R}^{nzf}$ in \eqref{eqt:ach_region_nzf} can be obtained by maximizing $R_2(\tri \tau)$ subject to a given value of $R_1(\tri \tau)$.
We begin by writing
\begin{equation}\label{eqt:sinr_nzf}
 \sinr_1(\tri \tau)=  2 \gamma - \frac{ 4 g_{12} \gamma^2 \cos^2(\phi_{12}+ \tri \tau)}{1+ 2 g_{12} \gamma}
\end{equation}
and obtain
\begin{equation}
 \tri \tau= \cos^{-1} \left( \sqrt{\frac{\left( 2 \gamma-\sinr_1\right) \left( 1+ 2 g_{12} \gamma\right)}{4 g_{12} \gamma^2}  }\right)- \phi_{12} + b \pi
\end{equation}
where the parameter $b$ takes either value 0 or 1. Denote $D=\frac{\left( 2 \gamma -\sinr_1\right) \left( 1+ 2 g_{12} \gamma\right)}{4 g_{12} \gamma^2}$ and
 $\tau'=\cos^{-1} \left(\sqrt{ D } \right)+ b \pi$. We substitute $\tri \tau= \tau' - \phi_{12}$ into $\sinr_2$ and the SINR of user $2$ can be written as
\begin{equation}\label{eqt:sinr_k}
\begin{aligned}
 \sinr_2&=2 \gamma- \frac{ 4 g_{21} \gamma^2 \cos^2\left(\phi_{21}-\tri \tau \right) }{1+ 2 g_{21} \gamma} \\
&=2 \gamma- \frac{ 4 g_{21} \gamma^2 \cos^2\left(\phi_{21}+\phi_{12} - \tau' \right) }{1+ 2 g_{21} \gamma}.
\end{aligned}
\end{equation}
Now we denote the sum of the channel phase rotation by
$\bar{\phi}= \phi_{21}+\phi_{12}
$ and use the following trigonometry properties:
$ \cos(x)= \cos(-x)$ and $\cos(\bar{ \phi}-\tau')= \cos(\bar{\phi})\cos(\tau')+\sin(\bar{\phi})\sin(\tau')$.
Since $\cos(\phi)=\sqrt{D}$, we have $\sin(\phi)=(-1)^b \sqrt{1-D}$.
Results follow by substituting trigonometric identities to \eqref{eqt:sinr_k}.

\section{Proof of Theorem \ref{thm:zf}}\label{app:lem_zf_rate}
\subsection{the maximum sum rate point}
Note that $\cos(a + \pi/2)=-\sin(a)$. The $\sinr$ pair can be written as the following:
$\sinr_1(\tri \tau)=\sin^2(\tri \tau +\phi_{12}) \gamma$ and $\sinr_2(\tri \tau)=\sin^2(\tri \tau - \phi_{21}) \gamma$.
Note that the sum rate is maximized if the product $f(\tri \tau)=(1+\sinr_1(\tri \tau))(1+ \sinr_2(\tri \tau))$ is maximized. 
Denote the derivative of $\sinr_i$ with respect to $\tri \tau$ as $\sinr'_i(\tri \tau)$. The KKT condition of the above maximization problem gives
\begin{equation}\label{eqt:kkt}
\begin{aligned}
f'(\tri \tau)&=\sinr'_1(\tri \tau) \left(1+ \sinr_2(\tri \tau) \right)\\
& \hspace{0.5cm} +\sinr'_2(\tri \tau) \left(1+ \sinr_1(\tri \tau) \right)=0.
\end{aligned}
\end{equation}
Computing the derivative of SINR's, we have 
$\sinr'_1(\tri \tau)= \sin(2\tri \tau +2\phi_{12}) \gamma$
and $\sinr'_2(\tri \tau)=\sin(2\tri \tau -2\phi_{21}) \gamma$.
\begin{figure*}[!t]
\begin{equation}\label{eqt:kkt_long}
 \begin{aligned}
   & \left( \sin(2\tri \tau +2\phi_{12})+ \sin(2\tri \tau -2\phi_{21}) \right) \gamma\\
   & + 2 \gamma^2 \left( \sin(\tri \tau -\phi_{21})\sin(\tri \tau + \phi_{12}) \right) 
\left(\cos(\tri \tau +\phi_{12})\sin(\tri \tau -\phi_{21})+\cos(\tri \tau -\phi_{21})\sin(\tri \tau + \phi_{12})  \right) =0\\
& \overset{(a)}{\Rightarrow} \left( \sin(2\tri \tau +2\phi_{12})+ \sin(2\tri \tau -2\phi_{21}) \right) \gamma
    + 2 \gamma^2 \left( \sin(\tri \tau -\phi_{21})\sin(\tri \tau + \phi_{12}) \right) 
\sin(2 \tri \tau + \tri \phi) =0 \\
& \overset{(b)}{\Rightarrow} 2 \left( \sin(2\tri \tau + \tri \phi) \cos(\bar{\phi}) \right) \gamma
    +  \gamma^2 \left( \cos(\bar{\phi})-\cos(2 \tri \tau + \tri \phi) \right) 
\sin(2 \tri \tau + \tri \phi) =0 \\
& \Rightarrow  \sin(2 \tri \tau + \tri \phi) \left( (2+\gamma) \cos(\bar{\phi}) - \gamma\cos(2 \tri \tau + \tri \phi) \right)=0, \pi\\
& \Rightarrow \tri \tau= -\half \tri \phi, \frac{\pi}{2}-\half \tri \phi \mbox{ or } \frac{1}{2} \cos^{-1}\left( \frac{2+\gamma}{\gamma} \cos(\bar{\phi})\right)- \frac{1}{2} \tri \phi
 \end{aligned}
\end{equation}
\hrule
\end{figure*}
With the fact that $\sin(2 \theta)=2 \sin(\theta) \cos(\theta)$, the KKT condition is then given by \eqref{eqt:kkt_long},
where $(a)$ is due to the fact that 
$ \sin(\theta_a)\cos(\theta_b)+\cos(\theta_a)\sin(\theta_b)=\sin(\theta_a+\theta_b)$ and $\tri \phi=\phi_{12}-\phi_{21}$. The transition $(b)$ is due to the fact that 
$  \sin(\theta_a)\sin(\theta_b)=\half \left(\cos(\theta_a-\theta_b)-\cos(\theta_a+\theta_b) \right)$ and
$\sin(\theta_a)+\sin(\theta_b)= 2 \sin\left( \frac{\theta_a+ \theta_b}{2}\right) \cos\left(\frac{\theta_a- \theta_b}{2} \right)$.

The second derivative of optimization objective $f(\tri \tau)$ is
\begin{equation}\label{eqt:sec_der}
\begin{aligned}
 f''(\tri \tau)&= \sinr''_1 \left( 1+ \sinr_2 \right) \\
& \hspace{0.5cm} + 2 \sinr'_1 \sinr'_2 \left( 1+ \sinr_1 \right) \sinr'_2
\end{aligned}
\end{equation} where $\sinr''_1$ and $\sinr''_2$ are the second derivative of $\sinr_1$ and $\sinr_2$ with respect to $\tri \tau$ respectively.
The argument $\tri \tau$ in the above equation is removed for the ease of notation. To check whether an extrema is a local maximum, one needs to substitute the solutions in \eqref{eqt:kkt_long} into \eqref{eqt:sec_der} and derive the conditions of which the second derivative is less than zero.
\begin{equation}
\begin{aligned}
& f''(\tri \tau)\leq 0 \mbox{ if }\\ 
& \left\{ \begin{array}{l}
                                           \tri \tau = -\half \tri \phi \mbox{ and } \cos(\bar{\phi})\leq \frac{\gamma}{2+\gamma}\\
\tri \tau = \frac{\pi}{2}-\half \tri \phi \mbox{ and } \cos(\bar{\phi})\geq -\frac{\gamma}{2+\gamma}\\
\tri \tau = \frac{1}{2} \cos^{-1}\left( \frac{2+\gamma}{\gamma} \cos(\bar{\phi})\right)- \frac{1}{2} \tri \phi \\
\hspace{2cm} \mbox{ and } \cos^2(\bar{\phi})\leq \frac{2\gamma^2}{2 \gamma^2 - (\gamma+2)^2}.
                                          \end{array}
\right.
\end{aligned}
\end{equation}
However, there may exist channel rotations $\bar{\phi}$ which satisfy more than one of the above conditions. In this case, there are two maximum points and one minimum point. To compute the maximum sum rate and its corresponding solution, one must compute the rates of the two maximum points for comparison.

\subsection{the max-min fairness solution}
For the max-min fair point $\max_{\tri \tau} \min(R_1(\tri \tau),R_2(\tri \tau))$, the optimal solution satisfies $R_1(\tri \tau)=R_2(\tri \tau)$. 
Then the optimal solution satisfies $\sin(\tri \tau + \phi_{12})= \pm \sin(\tri \tau - \phi_{21})$ and therefore $\tri \tau= -\half \tri \phi, \frac{\pi}{2}-\half \tri \phi$. This solution agrees with the max-min solution characterization in \cite{Park2010} when the channel model is in its standard form.

\section{Proof of Lemma \ref{lem:s0}}\label{app:s0}
We rewrite the rate function in Eqt. \eqref{eqt:rate_log_det} as a function of $\gamma, \bar{\mathbf{Q}}_1,\bar{\mathbf{Q}}_2$ where $\bar{\mathbf{Q}}_1=\frac{1}{P} \mathbf{Q}_1$. 
\begin{equation*}
\begin{aligned}
 R_1(\gamma, \bar{\mathbf{Q}}_1, \bar{\mathbf{Q}}_2) &= \half  \logdet \left(\half \mathbf{I}+ \gamma \left(  \bar{\mathbf{Q}}_1 
+ \mathbf{H}_{12} \bar{\mathbf{Q}}_2 \mathbf{H}_{12}^T\right) \right) \\
& -\half  \logdet \left(\half \mathbf{I}+ \gamma  \mathbf{H}_{12} \bar{\mathbf{Q}}_2 \mathbf{H}_{12}^T \right). 
\end{aligned}
\end{equation*}
Computing the derivative of $R_1(\gamma, \bar{\mathbf{Q}}_1, \bar{\mathbf{Q}}_2)$ gives the following,
\begin{equation*}\label{eqt:r1_diff}
 \begin{aligned}
&  \frac{2}{\ln(2)}\frac{\partial R_1(\gamma, \bar{\mathbf{Q}}_1, \bar{\mathbf{Q}}_2)}{\partial \gamma}\\
&=  \tr\left\{ \left(\half \mathbf{I} + \gamma\left( \bar{\mathbf{Q}}_1
+ \mathbf{H}_{12} \bar{\mathbf{Q}}_2 \mathbf{H}_{12}^T\right)\right)^{-1} \left( \bar{\mathbf{Q}}_1  
+ \mathbf{H}_{12} \bar{\mathbf{Q}}_2 \mathbf{H}_{12}^T \right)\right\}\\
& - \tr\left\{  \left(\half \mathbf{I} + \gamma \mathbf{H}_{12} \bar{\mathbf{Q}}_2 \mathbf{H}_{12}^T\right)^{-1} \mathbf{H}_{12} \bar{\mathbf{Q}}_2 \mathbf{H}_{12}^T \right\}.
 \end{aligned}
\end{equation*}
We obtain $\dot{R}_1=\left. \frac{\partial R_1(\gamma, \bar{\mathbf{Q}}_1, \bar{\mathbf{Q}}_2)}{\partial \gamma} \right|_{\gamma=0} = \ln(2) \cdot g_{11}$ which is independent to $\bar{\mathbf{Q}}_1,\bar{\mathbf{Q}}_2$.
Thus, by  Eqt. \eqref{eqt:s0}, $S_0$
is maximized if and only if $-\ddot{R}$ is minimized. In the following, we compute $\ddot{R}=\ddot{R}_1+\ddot{R}_2$ in terms of 
the transmit covariance matrices. It is sufficient to compute $\ddot{R}_1$ as $\ddot{R}_2$ can be computed by exchanging
the indexes. 
\begin{equation*}
\begin{aligned}
& \ddot{R}_1=\left. \frac{\partial^2 R_1(\gamma, \bar{\mathbf{Q}}_1, \bar{\mathbf{Q}}_2)}{\partial \gamma^2} \right|_{\gamma=0}\\
&= -2 \ln 2 \tr \left( \left( \bar{\mathbf{Q}}_1 + g_{12} \mathbf{J}(\phi_{12}) \bar{\mathbf{Q}}_2 \mathbf{J}(\phi_{12})^T \right)^2 \right) \\
& \hspace{3cm} + 2 g_{12}^2 \ln 2 \tr \left( \bar{\mathbf{Q}}_2^2 \right).
\end{aligned}
\end{equation*}
Thus, we have
\begin{eqnarray}\label{eqt:s0_max}
 & & \max_{\bar{\mathbf{Q}}_1,\bar{\mathbf{Q}}_2} S_0\\
\nonumber & \Leftrightarrow & \min_{\bar{\mathbf{Q}}_1,\bar{\mathbf{Q}}_2} - \ddot{R}_1 -\ddot{R}_2\\
\nonumber & \Leftrightarrow & \min_{\bar{\mathbf{Q}}_1,\bar{\mathbf{Q}}_2} \tr(\bar{\mathbf{Q}}_1^2) +\tr(\bar{\mathbf{Q}}_1^2)
 + 2 g_{12} \tr \left( \bar{\mathbf{Q}}_1 \mathbf{J}(\phi_{12}) \bar{\mathbf{Q}}_2 \mathbf{J}(\phi_{12})^T\right)\\
\nonumber & &\hspace{1cm} + 2 g_{21} \tr \left( \bar{\mathbf{Q}}_2 \mathbf{J}(\phi_{21}) \bar{\mathbf{Q}}_1 \mathbf{J}(\phi_{21})^T\right)
\end{eqnarray}
Let $\bar{\mathbf{Q}}_i= \mathbf{U}_i \diag(\lambda_i,1-\lambda_i) \mathbf{U}_i^T$ where $\mathbf{U}_i$ are orthogonal matrices.
Let $\mathbf{J}(\omega_1)=\mathbf{U}_1^T\mathbf{J}(\phi_{12}) \mathbf{U}_2$ and $\mathbf{J}(\omega_2)=\mathbf{U}_2^T\mathbf{J}(\phi_{21}) \mathbf{U}_1$
The maximization of spectral efficiency in Eqt. \eqref{eqt:s0_max} can be simplified to
\begin{equation}\label{eqt:so_max_1}
\begin{aligned}
& \min_{\omega_1,\omega_2,\lambda_1,\lambda_2} \sum_{i=1,2} \left( \lambda_i^2 + (1-\lambda_i)^2\right) \\
& + \sum_{i=1,2} 2 g_{ik} (2 \lambda_i \lambda_k -\lambda_i - \lambda_k +1) \cos(\omega_i)^2 \\
& + \sum_{i=1,2} 2 g_{ik} (\lambda_i + \lambda_k - 2 \lambda_i \lambda_k) \sin(\omega_i)^2
\end{aligned}
\end{equation} where $k=1$ if  $i=2$ and $k=2$ if $i=1$.
Setting the derivative of Eqt. \eqref{eqt:so_max_1} with respect to $\lambda_i$ to zero, we obtain the optimal solution, for $i=1,2$,
\begin{equation}\label{eqt:lambda}
\begin{aligned}
\lambda_i& = \half + \half g_{ik} (1-2 \lambda_k) \cos(2 \omega_i)\\
& \hspace{1cm} + \half g_{ki} (1- 2 \lambda_k) \cos(2 \omega_k).
\end{aligned}
\end{equation} 
Note that $\lambda_i$ is a function of $\lambda_k, \omega_i,\omega_k$. Surprisingly, substitute $\lambda_k$ into $\lambda_i$ and write
Eqt. \eqref{eqt:lambda} as a fixed point equation of $\lambda_i$ gives a unique solution independent of $\omega_1,\omega_2$:
$\lambda_i=\half$. $\lambda_k$ is obtained by plugging $\lambda_i=\half$ into Eqt. \eqref{eqt:lambda}:
\begin{equation}
 \lambda_1=\lambda_2=\half.
\end{equation} Hence, proper signals maximize the slope of spectral efficiency in the noise-limited regime.

\section{Proof of Lemma \ref{lem:power_offset}}\label{app:power_offset}
We define the following two terms
\begin{eqnarray}
  T_{11}&=&\lim_{\gamma \rightarrow \infty} \det \left( \half \mathbf{I} + \gamma \left( \mathbf{q}_1\mathbf{q}_1^T + \mathbf{H}_{12} \mathbf{q}_2 \mathbf{q}_2^T \mathbf{H}_{12}^T \right) \right)\\
\nonumber &=&\lim_{\gamma \rightarrow \infty} \gamma^2 \det \left( \mathbf{q}_1\mathbf{q}_1^T + \mathbf{H}_{12} \mathbf{q}_2 \mathbf{q}_2^T \mathbf{H}_{12}^T \right)\\
\nonumber T_{12}&=& \lim_{\gamma \rightarrow \infty} \det \left( \half \mathbf{I} + \gamma  \mathbf{H}_{12} \mathbf{q}_2 \mathbf{q}_2^T \mathbf{H}_{12}^T  \right)
=\lim_{\gamma \rightarrow \infty} \half \gamma g_{12}.
\end{eqnarray}
By exchanging the indices 1 and 2, we obtain $T_{21}$ and $T_{22}$ analogously.
By definition, 
\begin{eqnarray}
 \nonumber S_{\infty}(\mathbf{Q}) &=&\lim_{\gamma \rightarrow \infty} \frac{R_1(\gamma, \mathbf{q}_1 \mathbf{q}_1^T,\mathbf{q}_2 \mathbf{q}_2^T)+ R_2(\gamma, \mathbf{q}_1 \mathbf{q}_1^T,\mathbf{q}_2 \mathbf{q}_2^T)}{\log \gamma}\\
\nonumber &=& \lim_{\gamma \rightarrow \infty} \frac{1}{\log \gamma} \big(\half \sum_{i,j=1}^2 \log (T_{ij})\big)\\
&=& 1.
\end{eqnarray}
The high-SNR power offset for transmit beamforming vector $\mathbf{q}_1$ is
\begin{equation*}
 \begin{aligned}
&  L_{\infty}(\mathbf{Q})= \lim_{\gamma \rightarrow \infty} \left( \log \gamma - \frac{R_1(\mathbf{Q})+ R_2(\mathbf{Q})}{S_{\infty}}\right)\\
&=\lim_{\gamma \rightarrow \infty} \bigg( \log \gamma -\half \sum_{i,j=1}^2 \log (T_{ij}) \bigg)\\
&= -1 - \half \log \det\left( \half \mathbf{I} + \gamma  \mathbf{H}_{12} \mathbf{q}_2 \mathbf{q}_2^T \mathbf{H}_{12}^T  \right) + \half \log g_{12} \\
&  -\half \log \det\left( \half \mathbf{I} + \gamma  \mathbf{H}_{21} \mathbf{q}_1 \mathbf{q}_1^T \mathbf{H}_{21}^T  \right) + \half \log g_{21}\\
&\overset{(a)}{=} -1 - \half \log \sin^2 (\phi_{12}+ \tri \tau)- \half \log \sin^2 (\phi_{21}- \tri \tau).
 \end{aligned}
\end{equation*}
Recall that the transmit beamforming vectors are $\mathbf{q}_i=[\cos(\tau_i), \sin(\tau_i)]^T.$ The equation $(a)$ is due to the following. Note that $\mathbf{J}(\phi_{12})\mathbf{q}_2=[\cos(\phi_{12}+\tau_2), \sin(\phi_{12}+\tau_2)]^T$. Straightforward
computation gives $\det \left(\mathbf{q}_1 \mathbf{q}_1^T + \mathbf{H}_{12} \mathbf{q}_2 \mathbf{q}_2^T \mathbf{H}_{12}^T \right)=g_{12} \sin^2(\phi_{12}+\tau_2-\tau_1)$. Denote $\tri \tau=\tau_2-\tau_1$. 
Since the achievable rate of improper signaling performs better than the reference curve $\log \gamma$, the high-SNR power offset is negative. To improve efficiency, we \emph{minimize} the high-SNR power offset.
\begin{equation*}\label{eqt:power_offset_fin}
 \begin{aligned}
 & \min_{\tri \tau} L_{\infty}(\tri \tau) \\
&= \min_{\tri \tau}-1 - \half \log \sin^2 (\phi_{12}+ \tri \tau)- \half \log \sin^2 (\phi_{21}- \tri \tau)\\
& \Leftrightarrow  \min_{\tri \tau} |\sin(\phi_{12}+\tri \tau) \sin(\phi_{21}-\tri \tau) |\\
& \Leftrightarrow \min_{\tri \tau} |\cos(\phi_{12}-\phi_{21}+2 \tri \tau)-\cos(\phi_{12}+\phi_{21})|\\
& \Leftrightarrow \tri \tau = \left\{ \begin{array}{cc}
                     -\half \tri \phi & \mbox{ if } \cos(\bar{\phi})<0\\
		     \frac{\pi}{2} -\half \tri \phi & \mbox{ if } \cos(\bar{\phi})\geq 0
                    \end{array}
 \right.
 \end{aligned}
\end{equation*} where $\bar{\phi}=\phi_{12}+\phi_{21}$ and the optimized high-SNR power offset is
\begin{equation*}
 \max_{\tri \tau} L_{\infty}= \left\{ \begin{array}{cc}
                     -1 -\log \sin^2 (\half \bar{\phi}) & \mbox{ if } \cos(\bar{\phi})<0\\
		     -1 -\log \cos^2 (\half \bar{\phi}) & \mbox{ if } \cos(\bar{\phi})\geq 0.
                    \end{array}
 \right.
\end{equation*}

\bibliographystyle{IEEEbib}
\bibliography{bib6}

\section*{Acknowledgment}
The authors would like to thank Professor Erik Larsson for his insightful comments and discussion.

\begin{IEEEbiography}[{\includegraphics[width=1in,height=1.25in,clip,keepaspectratio]{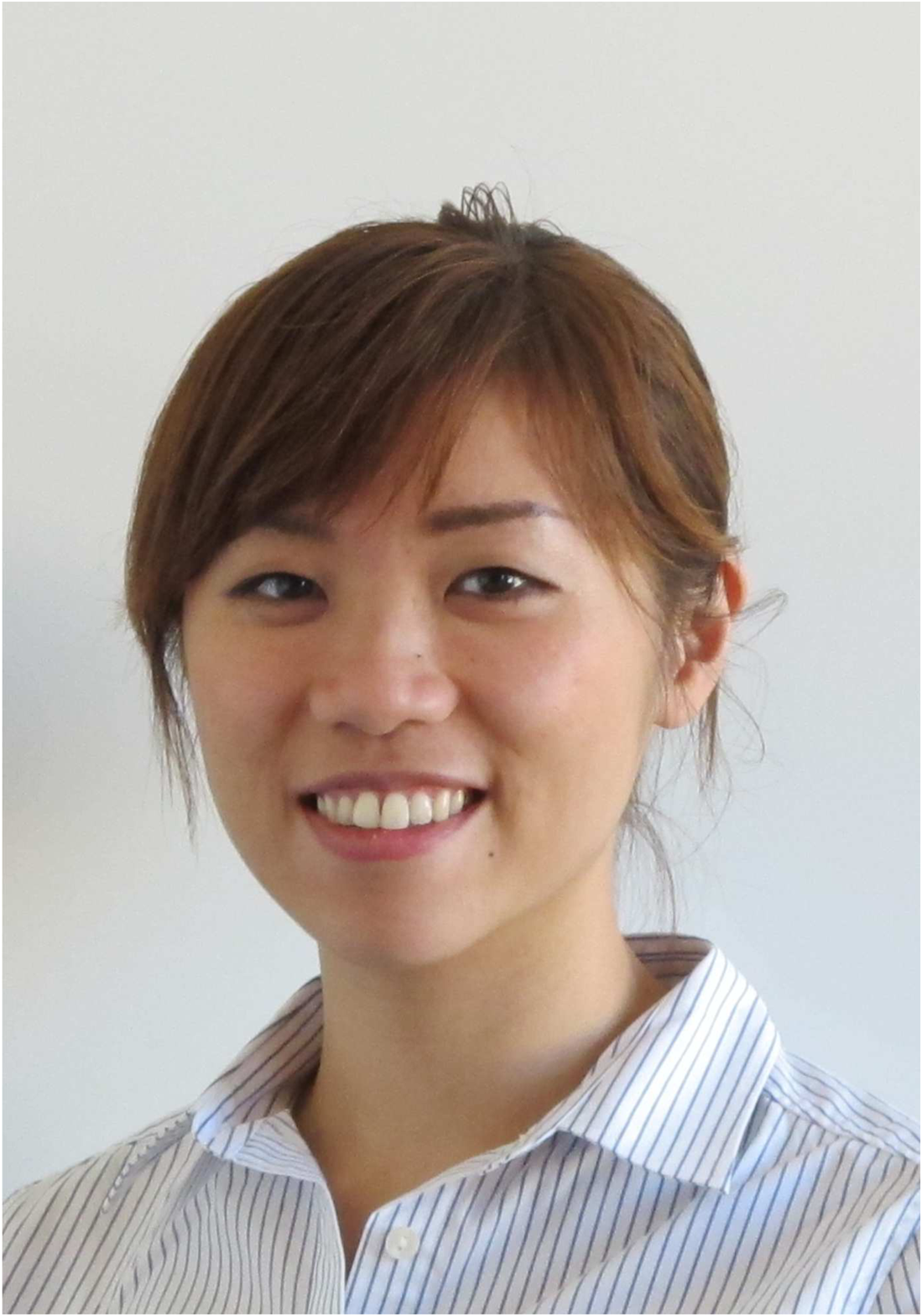}}]{Zuleita Ho}
Zuleita K. M. Ho (S'06-M'10) received her Ph.D. in wireless communication from EURECOM and Telecom Paris, France in 2010. In 2005 and 2007, She received her Bachelor and Master in Philosophy degree in Electronic Engineering (Wireless Communication) from Hong Kong University of Science and Technology (HKUST). From 2003 to 2004, she studied as a visiting student in Massachussets Institute of Technology (MIT) supported by the HSBC scholarship for Overseas Studies. From 2011 till now, she has joined the Communications Laboratory, chaired by Prof. Eduard Jorswieck, in the Department of Electrical Engineering and Information Technology at the Dresden University of Technology.
Zuleita enrolled into HKUST in 2002 through the Early Admission Scheme for gifted students. In 2003, she received the HSBC scholarship for Overseas Studies and visited MIT for 1 year. In 2007, she received one of the most prestigious scholarships in Hong Kong, The Croucher Foundation Scholarship , which supports her doctorate education in France. Other scholarships received include Sumida and Ichiro Yawata Foundation (2004, 2006), The Hong Kong Electric Co Ltd Scholarship (2004) and The IEE Outstanding Student Award (2004). 
\end{IEEEbiography}
\vspace*{-2\baselineskip}

\begin{IEEEbiography}[{\includegraphics[width=1in,height=1.25in,clip,keepaspectratio]{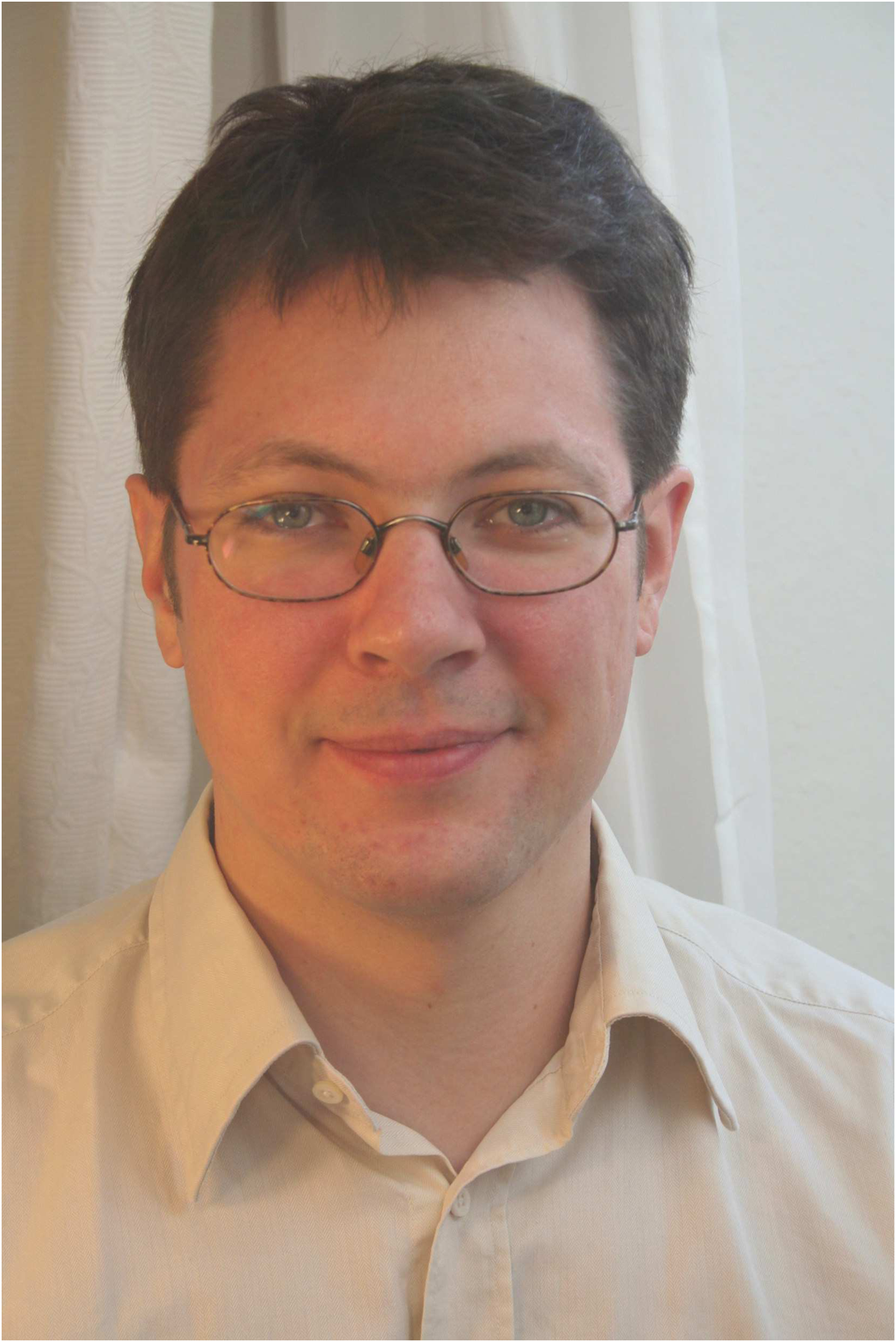}}]{Eduard Jorswieck}
Eduard A. Jorswieck (S’01-M’05-SM’08) received his Diplom-Ingenieur degree and Doktor-Ingenieur (Ph.D.) degree, both in electrical engineering and computer science from the Berlin University of Technology (TUB), Germany, in 2000 and 2004, respectively. He was with the Fraunhofer Institute for Telecommunications, Heinrich-Hertz-Institute (HHI) Berlin, from 2001 to 2006. In 2006, he joined the Signal Processing Department at the Royal Institute of Technology (KTH) as a post-doc and became a Assistant Professor in 2007. Since February 2008, he has been the head of the Chair of Communications Theory and Full Professor at Dresden University of Technology (TUD), Germany. His research interests are within the areas of applied information theory, signal processing and wireless communications. He is senior member of IEEE and elected member of the IEEE SPCOM Technical Committee. From 2008-2011 he served as an Associate Editor and since 2012 as a Senior Associate Editor for IEEE SIGNAL PROCESSING LETTERS. Since 2011 he serves as an Associate Editor for IEEE TRANSACTIONS ON SIGNAL PROCESSING. In 2006, he was co-recipient of the IEEE Signal Processing Society Best Paper Award. 
\end{IEEEbiography}
\vspace*{-2\baselineskip}

\end{document}